\documentclass[letterpaper]{article} 
\usepackage{aaai24}  
\usepackage{times}  
\usepackage{helvet}  
\usepackage{courier}  
\usepackage[hyphens]{url}  
\usepackage{graphicx} 
\usepackage{natbib}  
\usepackage{caption} 
\usepackage{algorithm}

\usepackage{newfloat}
\usepackage{listings}
\DeclareCaptionStyle{ruled}{labelfont=normalfont,labelsep=colon,strut=off} 
\floatstyle{ruled}
\newfloat{listing}{tb}{lst}{}
\floatname{listing}{Listing}
\usepackage{amsmath}
\usepackage{amssymb}
\usepackage{amsthm}
\usepackage{tikz}
\usepackage{pgfplots}
\usepackage{graphicx}
\usepackage{caption}
\usepackage{subcaption}
\usepackage{xcolor}
\allowdisplaybreaks  
\usepackage{cite}
\usepackage{algorithm}
\usepackage[noend]{algcompatible}
\usepackage{mathtools}
\usepackage{subcaption}
\usepackage{cleveref}
\usepackage{multirow}

\allowdisplaybreaks 

\newtheorem{theorem}{Theorem} 
\newtheorem{remark}{Remark}

\newtheorem{proposition}{Proposition}
\newtheorem{corollary}{Corollary}

\usepackage{comment}

\usepackage{url}

\title{Conformal Prediction Regions for Time Series\\ using Linear Complementarity Programming}
\author {
    Matthew Cleaveland\textsuperscript{\rm 1},
    Insup Lee\textsuperscript{\rm 2},
    George J. Pappas\textsuperscript{\rm 1}
    Lars Lindemann\textsuperscript{\rm 3}
}
\affiliations {
    \textsuperscript{\rm 1}Department of Electrical and Systems Engineering, University of Pennsylvania, Philadelphia, PA 19104, USA.\\
    \textsuperscript{\rm 2}Department of Computer and Information Science, University of Pennsylvania, Philadelphia, PA 19104, USA.\\
    \textsuperscript{\rm 3}Department of Computer Science, University of Southern California, Los Angeles, CA 90089, USA.
    mcleav@seas.upenn.edu, lee@seas.upenn.edu, pappasg@seas.upenn.edu, llindema@usc.edu
}

\begin{document}

\maketitle

\begin{abstract}
    Conformal prediction is a statistical tool for producing prediction regions of machine learning models that are valid with high probability. However, applying conformal prediction to  time series data leads to  conservative prediction regions. In fact,  to obtain prediction regions over $T$ time steps with confidence $1-\delta$, {previous works require that each individual prediction region is valid} with confidence $1-\delta/T$. We propose an optimization-based method for  reducing this conservatism to enable long horizon planning and verification when using learning-enabled time series predictors. Instead of considering prediction errors individually at each time step, we consider a parameterized prediction error over multiple time steps. By optimizing the parameters over an additional dataset, we find  prediction regions that are not conservative. We show that this problem can be cast as a mixed integer linear complementarity program (MILCP), which we then relax into a linear complementarity program (LCP). Additionally, we prove that the relaxed LP has the same optimal cost as the original MILCP. Finally, we demonstrate the efficacy of our method on case studies using pedestrian trajectory predictors and F16 fighter jet altitude predictors.
\end{abstract}


\section{Introduction}
\label{sec:intro}

Autonomous systems perform safety-critical tasks in dynamic  environments where system errors can be dangerous and costly, e.g., a mistake of a self-driving car navigating in urban traffic. In such scenarios, accurate uncertainty quantification is vital for ensuring safety of learning-enabled  systems. In this work, we focus on quantifying uncertainty of time series data. For instance, for a {given}  time series predictor that predicts the future trajectory of a pedestrian, how can we quantify the uncertainty and  accuracy of this predictor?

Conformal prediction (CP) has emerged as a popular method for statistical uncertainty quantification \cite{shafer2008tutorial,vovk2005algorithmic}. It aims to construct  prediction regions that contain the true quantity of interest with a user-defined probability. CP does not require any assumptions about the underlying distribution or the  predictor itself. Instead, one only {needs calibration data} that is exchangeable or independent and identically distributed. That means that CP can seamlessly be applied to learning-enabled predictors like neural networks, without the need to analyze the underlying architecture or change the training procedure \cite{angelopoulos2021gentle}.

One challenge of CP is that its standard variant can not directly be used to construct valid prediction regions for time series. This is because datapoints at different time steps are not exchangeable. Recently, variants of CP have been developed for time series, such as adaptive conformal inference \cite{gibbs2021adaptive,zaffran2022adaptive}. However, their coverage guarantees  are weaker and only asymptotic. While these types of guarantees are useful in many contexts, they are insufficient in  safety-critical applications such as self-driving cars.
\begin{figure}
    \centering
    \includegraphics[scale=0.45]{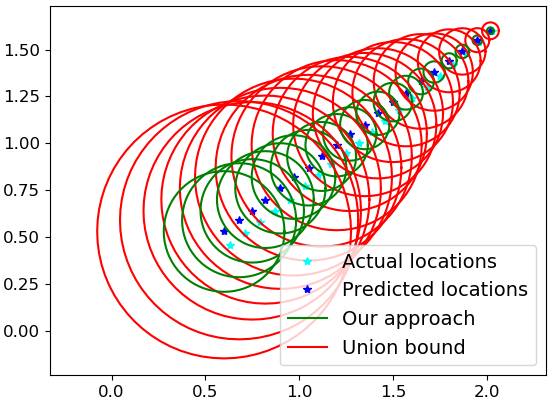}
        \caption{Sample trajectory of a pedestrian (teal stars), the pedestrian predictions using a social LSTM from \cite{alahi2016social} (blue stars), the conformal prediction regions for our approach (green circles), and the conformal prediction regions for the approach from \cite{lindemann2022safe} (red circles). Notably, our approach is less conservative.}
        \label{fig:ORCATrace}
\end{figure}

In previous works \cite{stankeviciute2021conformal,lindemann2022safe,lindemann2023conformal}, this issue was mitigated by instead considering a calibration dataset that consists of full rollouts of the time series, which do not violate the exchangeability assumption. {At design time, one runs a separate instance of CP for each of the $T$ time steps of the time series predictor (also referred to as a trajectory predictor). Then at runtime, these CP instances are combined with the predictor's $T$ predictions to compute prediction regions. }  However, to obtain
prediction regions that are valid over all $T$ time steps with confidence $1-\delta$, each individual prediction region has
to be valid with confidence $1 -\delta/T$.\footnote{A  union bound argument can be used to show this result.} This results in conservative prediction regions.

\textbf{Contributions: } To address this conservatism, this paper proposes constructing prediction regions for time series using linear complementary programming. Our main idea is based on i) mapping {trajectory prediction errors} to a single parameterized conformal scoring function to avoid union bounding, and ii) finding the optimal parameters from an additional {calibration} dataset. As opposed to existing works, our method enables non-conservative long horizon planning and verification. Our contributions are as follows:
\begin{itemize}
    \item We propose the conformal scoring function $R:=\max(\alpha_1R_1,\hdots,\alpha_TR_T)$ for time horizon $T$, prediction errors $R_t$, and parameters $\alpha_t$. Using $R$, we show that we can obtain non-conservative prediction regions   with $1-\delta$ confidence using conformal prediction.
    \item We formulate the problem of finding the parameters $\alpha_t$ that minimize $R$ from an additional calibration dataset as a mixed integer linear {complementarity} program, and we show how we can derive  a linear complementarity program that {has the same optimal value}.
    \item On two case studies, we demonstrate that our method produces much smaller valid conformal prediction regions compared to \cite{stankeviciute2021conformal,lindemann2022safe,lindemann2023conformal}, see Figure \ref{fig:ORCATrace}.
\end{itemize}

\subsection{Related Work}\label{sec:related} Conformal prediction was originally proposed by Vovk in \cite{vovk2005algorithmic,shafer2008tutorial} to quantify uncertainty of prediction models. The original method, however, required training a prediction model for each calibration datapoint, which is computationally intractable for learning-enabled predictors. Inductive conformal prediction, also referred to as split conformal prediction, addresses this issue by splitting the calibration data into two sets, one used for training the prediction model and one for applying conformal prediction  \cite{papadopoulos2008inductive}. Split conformal prediction has been extended to provide conditional probabilistic guarantees \cite{vovk2012conditional}, to handle distribution shifts \cite{tibshirani2019conformal,fannjiang2022conformal}, and to allow for quantile regression \cite{romano2019conformalized}.  Further, split conformal prediction has been used to construct probably approximately correct prediction sets for machine learning models \cite{park2020pac,angelopoulos2022conformal}, to perform out-of-distribution detection \cite{Kaur2022,kaur2023codit}, to guarantee safety in autonomous systems \cite{luo2022sample}, and to quantify uncertainty for F1/10 car motion predictions \cite{tumu2023physics}. Additionally, in \cite{stutz2022ICLR} the authors encode the width of the generated prediction sets directly into the loss function of a neural network during training.

However, the aforementioned methods cannot directly be applied to time series data. Multiple works have adapted conformal prediction algorithms for the time series domain, including enbpi \cite{xu2021conformal}, adaptive conformal inference  \cite{gibbs2021adaptive}, fully adaptive conformal inference  \cite{gibbs2022conformal,bhatnagar2023improved}, and copula conformal prediction \cite{sun2022copula,tonkens2023} . However, they either only provide averaged coverage guarantees over long horizons or require a lot of extra calibration data. Regardless, these methods have been used for analyzing financial market data \cite{Wisniewski2020ApplicationOC}, synthesizing safe robot controllers \cite{dixit2022}, and dynamically allocating compute resources \cite{cohen2023guaranteed}. For a more comprehensive overview of the conformal prediction landscape see \cite{angelopoulos2021gentle}.

Our paper is motivated by \cite{stankeviciute2021conformal,lindemann2022safe,lindemann2023conformal}. These works apply inductive conformal prediction in a time series setting, which requires constructing individual prediction regions with confidence $1-\delta/T$ to obtain $1-\delta$ coverage for all time steps. This results in unnecessarily conservative prediction regions {that are impractical for downstream tasks} such as motion planning. We overcome this limitation by considering a parameterized conformal scoring function that is defined over multiple time steps.


\section{Background}
\label{sec:background}

In this section, we present background on time series predictors and split conformal prediction for time series.

\subsection{Time Series Predictors}
\label{sec:backgroundTimeSeriesPredictors}

We want to predict  future values $Y_1,\hdots,Y_T$ of a time series for a prediction horizon of $T$ from past observed values $Y_{T_{obs}},\hdots, Y_0$ for an observation length of ${T_{obs}}$. Let  $\mathcal{D}$ denote an unknown distribution over a finite-horizon time series $Y:=(Y_{T_{obs}},\hdots, Y_0,Y_1,\hdots,Y_T)$, i.e., let $$(Y_{T_{obs}},\hdots, Y_0,Y_1,\hdots,Y_T) \sim \mathcal{D}$$ denote a random trajectory drawn from $\mathcal{D}$ where $Y_t \in \mathbb{R}^m$ is the value at time $t$. 
We make no assumptions about the distribution $\mathcal{D}$, but we do assume availability of a calibration dataset in which each trajectory is drawn independently from $\mathcal{D}$. Particularly, define $$ D_{cal}:= \{Y^{(1)},\hdots,Y^{(n)} \}$$
where the element $Y^{(i)}:=(Y_{T_{obs}}^{(i)},\hdots, Y_0^{(i)},Y_1^{(i)},\hdots,Y_T^{(i)})$ is independently drawn from $\mathcal{D}$, i.e., $Y^{(i)}\sim\mathcal{D}$.

A key challenge lies in constructing accurate predictions of future values $Y_1,\hdots,Y_T$ of the time series. We assume here that we are already given a time series predictor $h:\mathbb{R}^{m(T_{obs}+1)}\to \mathbb{R}^{mT}$ that maps ($Y_{T_{obs}},\hdots, Y_0$) to estimates of the next $T$ values ($Y_1,\hdots,Y_T$), denoted as $$(\hat{Y}_1,\hdots,\hat{Y}_T):=h(Y_{T_{obs}},\hdots, Y_0).$$

We make no assumptions about $h$. For example, it can be a long short-term memory network (LSTM) \cite{hochreiter1997long} or a sliding linear predictor with extended Kalman filters \cite{wei2022}. Using the predictor $h$, we can obtain predictions for the calibration data in $D_{cal}$. Specifically, for each $Y^{(i)}\in D_{cal}$, we obtain
\begin{align}\label{eq:predi}
    (\hat{Y}_1^{(i)},\hdots,\hat{Y}_T^{(i)}):=h(Y_{T_{obs}}^{(i)},\hdots, Y_0^{(i)}).
\end{align}

The predictor $h$ may not be exact, and {the inaccuracy of its predictions  $(\hat{Y}_1,\hdots,\hat{Y}_T)$ is unknown.} We will quantify prediction uncertainty using the calibration dataset $D_{cal}$ and split conformal prediction, which we introduce next.

\subsection{Split Conformal Prediction}
\label{sec:intro_cp}

 Conformal prediction was introduced in \cite{vovk2005algorithmic,shafer2008tutorial} to obtain valid prediction regions for complex predictive models such as neural networks without making assumptions on the  distribution of the underlying data. Split conformal prediction is a computationally tractable variant of conformal prediction  \cite{papadopoulos2008inductive} where it is assumed that a calibration dataset is available that has not been use to train the predictor. Let $R^{(0)},\hdots,R^{(k)}$ be $k+1$ exchangeable real-valued random variables.\footnote{Exchangeability is a slightly weaker form of independent and identically distributed (i.i.d.) random variables.} The variable $R^{(i)}$ is usually referred to as the \emph{nonconformity score}. In supervised learning, it is often defined as $R^{(i)}:=\|Z^{(i)}-\mu(X^{(i)})\|$ where the predictor $\mu$ attempts to predict the output  $Z^{(i)}$ based on the input $X^{(i)}$. Naturally, a large nonconformity score indicates a poor predictive model. Our goal is to obtain a prediction region for $R^{(0)}$ based on the calibration data $R^{(1)},\hdots,R^{(k)}$, i.e., the random variable $R^{(0)}$ should be contained within the prediction region  with high probability. 

Formally, given a failure probability $\bar{\delta}\in (0,1)$, we want to construct a valid prediction region $C\in \mathbb{R}$ so that\footnote{More formally, we would have to write $C(R^{(1)},\hdots,R^{(k)})$ as the prediction region $C$ is a function of $R^{(1)},\hdots,R^{(k)}$. For this reason, the probability measure $P$ is defined over the product measure of $R^{(0)},\hdots,R^{(k)}$.}
\begin{align}\label{eq:cpGuaranteeVanilla}
    \text{Prob}(R^{(0)}\le C)\ge 1-\bar{\delta}.
\end{align}

We pick $C:=\text{Quantile}(\{ R^{(1)}, \hdots, R^{(k)}, \infty \},1-\bar{\delta})$ which is  the $(1-\bar{\delta})$th quantile of the empirical distribution of the values $R^{(1)},\hdots,R^{(k)}$ and $\infty$. Equivalently, by assuming that $R^{(1)},\hdots,R^{(k)}$ are sorted in non-decreasing order and by adding $R^{(k+1)}:=\infty$, we can  obtain $C:=R^{(p)}$ where $p:=\lceil (k+1)(1-\bar{\delta})\rceil$ with $\lceil \cdot\rceil$ being the ceiling function, i.e., $C$ is the $p$th smallest nonconformity score. By a quantile argument, see \cite[Lemma 1]{tibshirani2019conformal}, one can prove that this choice of $C$ satisfies~\Cref{eq:cpGuaranteeVanilla}. We remark that $k\ge \lceil (k+1)(1-\bar{\delta})\rceil$ is required to hold to obtain meaningful, i.e., bounded, prediction regions. It is known that the guarantees in \eqref{eq:cpGuaranteeVanilla} are marginal over the randomness in $R^{(0)}, R^{(1)}, \hdots, R^{(k)}$ as opposed to being conditional on $ R^{(1)}, \hdots, R^{(k)}$. In fact, the probability $\text{Prob}(R^{(0)}\le C|R^{(1)}, \hdots, R^{(k)})$ is a random variable that follows a beta distribution centered around $1-\bar{\delta}$ with decreasing variance as $k$ increases, see \cite[Section 3.2]{angelopoulos2021gentle} for details. As a result, larger calibration datasets 
reduce the variance of  conditional coverage.

\begin{remark}\label{rem:1}
Note that  split conformal prediction assumes that the non-conformity scores $R^{(0)},\hdots,R^{(k)}$ are exchangeable. This complicates its use for time series data where $Y_{t+1}$ generally depends on $Y_t$. However, with a calibration dataset $D_{cal}$ of trajectories independently drawn from $\mathcal{D}$, the application of conformal prediction is possible \cite{stankeviciute2021conformal,lindemann2022safe,lindemann2023conformal}. In summary, one defines $T$ non-conformity scores $R^{(i)}_t:=\|\hat{Y}^{(i)}_t-Y^{(i)}_t\|$ for each time $t\in\{1,\hdots,T\}$ and each calibration trajectory $Y^{(i)}\in D_{cal}$; $R^{(i)}_t$ can be interpreted as the $t$-step ahead prediction error. Using conformal prediction, for each time $t$ we can then construct a prediction region $C_t$ so that $\text{Prob}(\|\hat{Y}_t-Y_t\|\le C_t)\ge 1-\bar{\delta}$. To then obtain a prediction region over all time steps, i.e, $\text{Prob}(\|\hat{Y}_t-Y_t\|\le C_t, \; \forall t\in \{1,\hdots,T\})\ge 1-\delta$, we simply set $\bar{\delta}:=\delta/T$. As previously mentioned, this results in conservative prediction regions that are not practical.
\end{remark}


\section{Problem Formulation and Main Idea}
\label{sec:problemFormulation}

Our goal is to construct valid conformal prediction regions for time series. The main idea is to use a parameterized non-conformity score that considers prediction errors over all $T$ time steps. We consider the non-conformity score
\begin{align}
    R:=\max(\alpha_1 R_p(Y_1, \hat{Y}_1), \hdots, \alpha_T R_p(Y_T,\hat{Y}_T)) \label{eq:generalRfunc}
\end{align}
where $\alpha_1,\hdots,\alpha_T\geq 0$ are parameters and where $R_p$ is the prediction error for an individual time step. For the purpose of this paper, we use the Euclidean distance between the actual and the predicted value of the time series at time $t$ as
\begin{align}
    R_p(Y_t, \hat{Y}_t) := \| Y_t - \hat{Y}_t \| \label{eq:RfuncL2}
\end{align}
We note that $R_p$ can in general be any real-valued function of $Y_t$ and $\hat{Y}_t$. In this paper, however, we set $R_p$ to be the Euclidean distance between $Y_t$ and $\hat{Y}_t$ as in \eqref{eq:RfuncL2}.

Taking the maximum prediction error over all $T$ time steps as in equation \eqref{eq:RfuncL2} will allow us to avoid the potentially conservative prediction regions from \cite{stankeviciute2021conformal,lindemann2022safe,lindemann2023conformal} as we, in our proof, do not need to union bound over the individual prediction errors. However, the choice of the parameters $\alpha_1,\hdots,\alpha_T$ is now important. If we apply conformal prediction as described in \Cref{sec:intro_cp}  to the nonconformity score in \eqref{eq:generalRfunc} with $\alpha_1=\hdots=\alpha_T=1$, this would typically result in large prediction regions for the first time step since prediction errors usually get larger over time, i.e., $R_p(Y_T,\hat{Y}_T)$  in \eqref{eq:generalRfunc} would dominate so that  $R=R_p(Y_T,\hat{Y}_T)$. The $\alpha_1,\hdots,\alpha_T$ parameters serve to weigh the different time steps, and we generally expect $\alpha_t$ for larger times to be smaller than $\alpha_t$ for smaller times. 

Given the parameterized function $R$ as in \Cref{eq:generalRfunc} and a failure threshold of $\delta\in (0,1)$, we are in this paper particularly interested in computing the parameters $\alpha_1,\hdots,\alpha_T$ that minimize the constant $C$ that satisfies 
    \begin{align*}
        \text{Prob}(R\le C)\ge 1-\delta.
    \end{align*}

For given parameters $\alpha_1,\hdots,\alpha_T$, note that the constant $C$ can be found by applying conformal prediction to the nonconformity score $R$ using the calibration set $D_{cal}$. When we have found this constant $C$, we  know that 
\begin{align*}
    \text{Prob}(\|Y_t-\hat{Y}_t\|\le C/\alpha_t, \;\; \forall t\in\{1,\hdots,T\})\ge 1-\delta.
\end{align*}

To solve the aforementioned problem, we split the calibration dataset $D_{cal}$ into two sets $D_{cal,1}$ and $D_{cal,2}$. We will use the first calibration dataset $D_{cal,1}$ to compute the values of $\alpha_1,\hdots,\alpha_T$ that give the smallest prediction region $C$ for $R$ (see Section \ref{sec:approach}). In a next step, we use these parameters that now fully define the nonconformity score $R$ along with the second calibration dataset $D_{cal,2}$ to construct valid prediction regions via  conformal prediction (see Section \ref{sec:second_approach}).


\section{Nonconformity Scores for Time-Series via Linear Programming}
\label{sec:approach}

We first have to set up some notation. Let us denote the elements of the first calibration dataset as $D_{cal,1}:=\{Y^{(1)},\hdots,Y^{(n_1)}\}$ where $n_1>0$ is the size of the first calibration dataset. Recall that each element $Y^{(i)}$ of $D_{cal,1}$ is defined as $Y^{(i)}=(Y_{T_{obs}}^{(i)},\hdots,Y_0^{(i)},Y^{(i)}_1,\hdots,Y^{(i)}_T)$. Using the trajectory predictor $h$, we obtain predictions  $(\hat{Y}_1^{(i)},\hdots,\hat{Y}_T^{(i)})$ for $(Y^{(i)}_1,\hdots,Y^{(i)}_T)$ as per equation \eqref{eq:predi}.

We then compute the prediction error $R_{t}^{(i)}$ according to equation \eqref{eq:RfuncL2} for each calibration trajectory and for each time step, i.e., for each $t\in\{1,\hdots,T\}$ and $i\in\{1,\hdots,n_1\}$. More formally,  we compute for each calibration trajectory
\begin{align}\label{eq:non_pred_err}
    R_t^{(i)}:=R_p(Y_t^{(i)},\hat{Y}_t^{(i)}).
\end{align}

We can now cast the problem of finding the parameters $\alpha_1,\hdots,\alpha_T$  as the following optimization problem.

\begin{subequations}\label{eq:highLevelOptimization_}
\begin{align} 
    \min_{ \alpha_1,\hdots,\alpha_T} &\quad \text{Quantile}(\{ R^{(1)}, \hdots, R^{(n_1)} \},1-\delta) \label{eq:highLevelOptimization} \\
    \text{s.t.} & \quad R^{(i)} = \text{max}(\alpha_1 R^{(i)}_1, \hdots, \alpha_T R^{(i)}_T), i=1,\hdots,n_1 \label{eq:highLevelOptimizationb}  \\
    & \quad \sum_{j=1}^{T} \alpha_j = 1 \label{eq:highLevelOptimizationc}  \\
    & \quad \alpha_j \geq 0, j=1,\hdots,T \label{eq:highLevelOptimizationd}
\end{align}
\end{subequations}
where $\text{Quantile}(\{ R^{(1)}, \hdots, R^{(n_1)} \},1-\delta)$ denotes the empirical $1-\delta$ quantile over  $ R^{(1)}, \hdots, R^{(n_1)}$. We  removed the value $\infty$ from \text{Quantile}, as required in Section \ref{sec:intro_cp}. However, this only affects the computation of $\alpha_1,\hdots,\alpha_T$, but not the correctness of our method as shown later in Section \ref{sec:second_approach}. 
Note also that one could alternatively use $\text{Quantile}(\{ R^{(1)},\hdots,R^{(n_1)} \},(1+1/{n_1})(1-\delta))$ as it holds that $\text{Quantile}(\{ R^{(1)},\hdots,R^{(n_1)},\infty\},1-\delta) = \text{Quantile}(\{ R^{(1)},\hdots,R^{(n_1)} \},(1+1/{n_1})(1-\delta))$ if $(1+1/{n_1})(1-\delta)\in (0,1)$. 
Finally, note that Constraint \eqref{eq:highLevelOptimizationd} constrains $\alpha_1,\hdots,\alpha_T$ to be non-negative. However, one of the $\alpha_t$ could be zero which would be a problem. \Cref{eq:highLevelOptimizationc}  normalizes $\alpha_1,\hdots,\alpha_T$ so that this does not happen. Particularly, we show next that the optimal values  $\alpha_1,\hdots,\alpha_T$ of  \eqref{eq:highLevelOptimization_}  will be positive as long as all $R_{t}^{(i)}$ are non-zero. 
\begin{theorem}\label{thm:alphasNonZero}
    Assume that $R_{t}^{(i)}>0$ for all $t\in\{1,\hdots,T\}$ and for all $i\in\{1,\hdots,n_1\}$. Then, the optimal values $\alpha_1,\hdots,\alpha_T$ of \eqref{eq:highLevelOptimization_} are positive, i.e.,  $\alpha_1^*,\hdots,\alpha_T^* > 0$.
\end{theorem}
The proof of this theorem, along with subsequent theorems, can be found in the Appendix.

Note that the intuition behind this result is that \Cref{eq:highLevelOptimizationc} assigns a budget for the parameters $\alpha_1,\hdots,\alpha_T$ that need to sum to one. If  $\alpha_t=0$, one can increase $\alpha_t$ to allow another $\alpha_\tau$ with $\tau\neq t$ to decrease which will result in lowering the cost in equation \eqref{eq:highLevelOptimization}. We note that the assumption of having non-zero values of $R_{t}^{(i)}$ is not limiting as in practice prediction errors are never zero. 

\begin{remark}
    We note that the optimization problem in equation \eqref{eq:highLevelOptimization_} is always feasible. This follows since any set of parameters $\alpha_1,\hdots,\alpha_T$ that satisfy \Cref{eq:highLevelOptimizationc,eq:highLevelOptimizationd} constitute a feasible solution to \eqref{eq:highLevelOptimization_}.
\end{remark}

To solve the optimization problem in equation \eqref{eq:highLevelOptimization_}, we transform it into a linear complementarity program \cite{cottle2009linear} that has the same optimal value as \eqref{eq:highLevelOptimization_}. To do so, we first reformulate the $\text{max}()$ operator in equation \eqref{eq:highLevelOptimizationb} as a set of mixed integer linear program (MILP) constraints, which we then {transform} into a set of linear program constraints. We show that this transformation preserves the optimal value of \eqref{eq:highLevelOptimization_} (Section \ref{sec:maxFormulation}). Then, we reformulate the Quantile function from equation \eqref{eq:highLevelOptimization} into a linear program (LP) that we  reformulate using its KKT conditions (Section \ref{sec:quantileKKT}), resulting in a linear complementarity program.

\subsection{Reformulating the Max Operator}
\label{sec:maxFormulation}

We first use ideas from \cite{bemporad1999,raman2014model} for encoding the max operator in \eqref{eq:highLevelOptimizationb} as an MILP. Therefore, we introduce the binary variables $b_t^{(i)}\in \{ 0,1\}$ for each time $t\in\{1,\hdots,T\}$ and calibration trajectory $i\in\{1,\hdots,n_1\}$. The equality constraint $R^{(i)} = \text{max}(\alpha_1 R_1^{(i)}, \hdots, \alpha_T R_T^{(i)})$ in equation \eqref{eq:highLevelOptimizationb} is then equivalent to the following MILP constraints:
\begin{subequations} \label{eq:maxConstraints}
\begin{align}
    & \quad R^{(i)} \geq \alpha_t R_t^{(i)}, \; t=1,\hdots,T \\
    & \quad R^{(i)} \leq \alpha_t R_t^{(i)} + (1-b_t^{(i)}) M, \; t=1,\hdots,T  \\
    & \quad \sum_{t=1}^{T} b_t^{(i)} = 1  \\
    & \quad b_t^{(i)} \in \{ 0,1\}, \; t=1,\hdots,T 
\end{align}
\end{subequations}
where $M>0$ is a sufficiently large and positive constant.\footnote{See \cite{bemporad1999}. In this case, we need $M$ to satisfy  $M\geq \max_{i\in\{1,\hdots,n_1\},t\in\{1,\hdots,T\}}R_t^{(i)}$.}

By replacing the max operator in \Cref{eq:highLevelOptimizationb} with the MILP encoding in \Cref{eq:maxConstraints}, we arrive at the following optimization problem: 
\begin{subequations}\label{eq:highLevelOptimizationMaxMILP}
\begin{align}
    \min_{\alpha_t,b_{t}^{(i)},R^{(i)}} &\quad \text{Quantile}(\{ R^{(1)}, \hdots, R^{(n_1)} \},1-\delta) \label{eq:highLevelOptimizationMaxMILPa} \\
    \text{s.t.} & \quad R^{(i)} \geq \alpha_t R_t^{(i)}, \; i=1,\hdots,n_1 \land t=1,\hdots,T \label{eq:highLevelOptimizationMaxMILPb}  \\
    & \quad R^{(i)} \leq \alpha_t R_t^{(i)} + (1-b_t^{(i)}) M, \; i=1,\hdots,n_1 \label{eq:highLevelOptimizationMaxMILPc}  \\ 
    & \quad \quad \quad \quad \quad \quad \quad \quad \quad \quad \quad \quad \; \; \land t=1,\hdots,T \nonumber  \\
    & \quad \sum_{t=1}^{T} b_t^{(i)} = 1 \label{eq:highLevelOptimizationMaxMILPd}  \\
    & \quad b_t^{(i)} \in \{ 0,1\}, \; t=1,\hdots,T \label{eq:highLevelOptimizationMaxMILPe} \\
    & \quad \eqref{eq:highLevelOptimizationc}, \; \eqref{eq:highLevelOptimizationd} \nonumber  
\end{align}
\end{subequations}
The constraints in \eqref{eq:highLevelOptimizationMaxMILP} are linear in their parameters, i.e., in $\alpha_t$, $R^{(i)}$, and $b_t^{(i)}$. The cost function in \eqref{eq:highLevelOptimizationMaxMILPa} {can be reformulated as a linear complementarity program (LCP)}, as shown in the next section. Thus, \eqref{eq:highLevelOptimizationMaxMILP} is an MILCP since the $b_t^{(i)}$ variables are binary due to equation \eqref{eq:highLevelOptimizationMaxMILPe}. Commercial solvers can solve MILCPs, but they present scalability issues for larger time horizons $T$ and calibration data $n_1$. Therefore, we present an LCP reformulation of \eqref{eq:highLevelOptimizationMaxMILP} which will achieve the same optimality in terms of the constraint function \eqref{eq:highLevelOptimizationMaxMILPa} {while providing better scalability}.

\subsubsection{Linear {complementarity relaxation} of the MILCP in \eqref{eq:highLevelOptimizationMaxMILP}}
\label{sec:relaxMax}

We first remove the upper bound on $R^{(i)}$ in equation \eqref{eq:highLevelOptimizationMaxMILPc} and the binary variables $b_t^{(i)}$ in equations \eqref{eq:highLevelOptimizationMaxMILPd} and \eqref{eq:highLevelOptimizationMaxMILPe} from the MILCP in \eqref{eq:highLevelOptimizationMaxMILP}. This results in the optimization problem: 
\begin{subequations} \label{eq:problemWithoutBinaryVars}
\begin{align} 
    \min_{\alpha_t,R^{(i)}} &\quad \text{Quantile}(\{ R^{(1)}, \hdots, R^{(n_1)} \},1-\delta) \label{eq:problemWithoutBinaryVarsa} \\
    \text{s.t.} & \quad R^{(i)} \geq \alpha_t R_t^{(i)}, \; i=1,\hdots,n_1 \land t=1,\hdots,T \label{eq:problemWithoutBinaryVarsb}  \\
    & \quad \eqref{eq:highLevelOptimizationc}, \; \eqref{eq:highLevelOptimizationd} \nonumber 
\end{align}
\end{subequations}
Note that the above optimization problem is now a linear {complementarity} program due to the removal of the binary variables $b_t^{(i)}$. While the LCP in \eqref{eq:problemWithoutBinaryVars} is clearly computationally more tractable than the MILCP in \eqref{eq:highLevelOptimizationMaxMILP}, one may ask in what way the optimal solutions of these two optimization problems are related. We next prove that the optimal values (in terms of the cost functions in \eqref{eq:highLevelOptimizationMaxMILPa} and \eqref{eq:problemWithoutBinaryVarsa}) of the two optimization problems \eqref{eq:highLevelOptimizationMaxMILP} and \eqref{eq:problemWithoutBinaryVars} are the same.

\begin{theorem}\label{thm:1}
    Let $\delta \in (0,1)$ be a failure probability and $R_t^{(i)}$ be the prediction errors as in \eqref{eq:non_pred_err} for times $t\in\{1,\hdots,T\}$ and calibration trajectories $i\in\{1,\hdots,n_1\}$ from $D_{cal,1}$. Then, the optimal {cost} values of the optimization problems \eqref{eq:highLevelOptimizationMaxMILP} and \eqref{eq:problemWithoutBinaryVars} are equivalent.
\end{theorem}


Note that Theorem \ref{thm:1} guarantees that the optimal values of the optimization problems in \eqref{eq:highLevelOptimizationMaxMILP} and \eqref{eq:problemWithoutBinaryVars} are equivalent. However, we do not guarantee that the optimal parameters $\alpha_1^*,\hdots,\alpha_T^*$ obtained by solving \eqref{eq:highLevelOptimizationMaxMILP} and \eqref{eq:problemWithoutBinaryVars} will be the same as the optimization problem in \eqref{eq:highLevelOptimizationMaxMILP} may have multiple optimal solutions $\alpha_1^*,\hdots,\alpha_T^*$. Nonetheless, the linear reformulation in \eqref{eq:problemWithoutBinaryVars} is computationally tractable and cost optimal. In our experiments in \Cref{sec:alphaDiffs}, we also show that we obtain parameters from \eqref{eq:highLevelOptimizationMaxMILP} and \eqref{eq:problemWithoutBinaryVars} that are almost equivalent.

\subsection{Reformulating the Quantile Function}
\label{sec:quantileKKT}

We note that the Quantile function in \eqref{eq:highLevelOptimization} can be written as the following linear program \cite{Koenker1978}:
\begin{subequations}\label{eq:quantileAsLP}
\begin{align} 
    \text{Quantile}&(\{ R^{(1)}, \hdots, R^{(n_1)} \},1-\delta) \nonumber \\ 
    = \text{argmin}_{q} &\quad \sum_{i=1}^{n_1} \left( (1-\delta) e_i^+ + \delta e_i^- \right) \label{eq:quantileAsLPa}  \\
    \text{s.t.} & \quad e_i^+-e_i^- = R^{(i)} - q, \; i=1,\hdots,n_1 \label{eq:quantileAsLPb}  \\
    & \quad e_i^+,e_i^- \geq 0, \; i=1,\hdots,n_1 \label{eq:quantileAsLPc}
\end{align}
\end{subequations}
where the optimal solution $q$ of \eqref{eq:quantileAsLP} is equivalent to $\text{Quantile}(\{ R^{(1)}, \hdots, R^{(n_1)} \},1-\delta)$. However, we cannot directly replace $\text{Quantile}(\{ R^{(1)}, \hdots, R^{(n_1)} \},1-\delta)$ in the optimization problem in \eqref{eq:highLevelOptimization} with the linear program from equation \eqref{eq:quantileAsLP} as that would result in a problem with an objective function that is itself an optimization problem. To address this issue, we replace the optimization problem in \eqref{eq:quantileAsLP} with its KKT conditions, which we provide in equation \eqref{eq:KKTconds} in Appendix \ref{app:KKTConds} due to space limitations.

The KKT conditions form an LCP because every constraint is either linear (\eqref{eq:KKTcondsa}-\eqref{eq:KKTcondsc} and \eqref{eq:KKTcondsf}-\eqref{eq:KKTcondsj}) or an equality constraint stating the multiple of two positive variables equals $0$ (\eqref{eq:KKTcondsd} and \eqref{eq:KKTcondse}).\footnote{These constraints ensure that one of the variables equals $0$.} Despite the slight non-linearity, LCPs can be efficiently solved with existing tools \cite{yu2019solving}. Because feasible linear optimization problems have zero duality gap \cite{bradley1977applied}, any solution to the KKT conditions in \eqref{eq:KKTconds} is also an optimal solution to Problem \eqref{eq:quantileAsLP}. 

\subsection{Summarizing the Optimization Programs}
\label{sec:finalOptimizationProblem}

Based on the results from the previous subsections, we can formulate two optimization problems that solve \eqref{eq:highLevelOptimization_}.

First, we can exactly solve \eqref{eq:highLevelOptimization_} by replacing the max operator in \eqref{eq:highLevelOptimizationb} with the MILP in \eqref{eq:maxConstraints} and the quantile function in \eqref{eq:highLevelOptimization} with its KKT conditions in \eqref{eq:KKTconds}. The following result follows by construction (for sufficiently large $M$).
\begin{proposition}
The MILCP
\begin{subequations}\label{eq:originalProblemMILP}
\begin{align}
    \min_{q,\alpha_t,R^{(i)},b_t^{(i)}} &\quad q \\
    \text{s.t.} & \quad \eqref{eq:highLevelOptimizationMaxMILPb}, \; \eqref{eq:highLevelOptimizationMaxMILPc}, \; \eqref{eq:highLevelOptimizationMaxMILPd}, \; \eqref{eq:highLevelOptimizationMaxMILPe} \nonumber \\
    & \quad \eqref{eq:KKTcondsa}, \; \eqref{eq:KKTcondsb}, \; \eqref{eq:KKTcondsc}, \; \eqref{eq:KKTcondsd}, \; \eqref{eq:KKTcondse} \nonumber \\ 
    & \quad \eqref{eq:KKTcondsf}, \; \eqref{eq:KKTcondsg}, \; \eqref{eq:KKTcondsh}, \; \eqref{eq:KKTcondsi}, \; \eqref{eq:KKTcondsj}  \nonumber \\
    & \quad \eqref{eq:highLevelOptimizationc}, \; \eqref{eq:highLevelOptimizationd} \nonumber  
\end{align}
is equivalent to the optimization problem in \eqref{eq:highLevelOptimization_}.
\end{subequations}
\end{proposition}

Second, by further removing the upper bound from the max operator as described in \Cref{sec:relaxMax}, we can solve a linear complementarity program that has the same optimal {cost} value as \eqref{eq:highLevelOptimization_}. This result follows from Theorem \ref{thm:2}.
\begin{proposition}
The LCP 
\begin{subequations}\label{eq:relaxedProblemLP}
\begin{align}
    \min_{q,\alpha_t,R^{(i)}} &\quad q \\
    \text{s.t.} & \quad \eqref{eq:highLevelOptimizationMaxMILPb} \nonumber \\
    & \quad \eqref{eq:KKTcondsa}, \; \eqref{eq:KKTcondsb}, \; \eqref{eq:KKTcondsc}, \; \eqref{eq:KKTcondsd}, \; \eqref{eq:KKTcondse} \nonumber \\ 
    & \quad \eqref{eq:KKTcondsf}, \; \eqref{eq:KKTcondsg}, \; \eqref{eq:KKTcondsh}, \; \eqref{eq:KKTcondsi}, \; \eqref{eq:KKTcondsj}  \nonumber \\
    & \quad \eqref{eq:highLevelOptimizationc}, \; \eqref{eq:highLevelOptimizationd} \nonumber  
\end{align}
\end{subequations}
has the same optimal value as the optimization problem  \eqref{eq:highLevelOptimization_}.
\end{proposition}

\section{Conformal Prediction Regions for Time-Series}
\label{sec:second_approach}

Finally, after obtaining the parameters $\alpha_1,\hdots,\alpha_T$ from the first calibration dataset $D_{cal,1}$ by solving \eqref{eq:originalProblemMILP} or \eqref{eq:relaxedProblemLP}, we can apply conformal prediction with the non-conformity score $R$ as defined in \eqref{eq:generalRfunc} to obtain prediction regions for time series. We do so by following  \Cref{sec:intro_cp} and using the second calibration dataset $D_{cal,2}$.

More formally, let us denote the elements of the second calibration dataset as $D_{cal,2}:=\{Y^{(n_1+1)},\hdots,Y^{(n)}\}$ where $n-n_1>0$ is the size of the second calibration dataset.\footnote{Recall the $n$ is the size of $D_{cal}$ and $n_1$ is the size of $D_{cal,1}$.} For each calibration trajectory $i\in\{n_1+1,\hdots,n\}$, we compute the nonconformity score as
\begin{align}\label{eq:cal_R_}
    R^{(i)}:=\max\left(\alpha_1 R_p(Y_1^{(i)}, \hat{Y}_1^{(i)}), \hdots, \alpha_T R_p(Y_T^{(i)},\hat{Y}_T^{(i)})\right).
\end{align}
We are now ready to state the final result of our paper in which we obtain prediction regions for $R:=\max\left(\alpha_1 R_p(y_1, \hat{y}_1), \hdots, \alpha_T R_p(y_T,\hat{y}_T)\right)$.
\begin{theorem}\label{thm:2}
Let $Y$ be a trajectory drawn from $\mathcal{D}$ and let $\hat{Y}=h(Y_{T_{obs}},\hdots,Y_0)$ be predictions from the time series predictor $h$. Let $\delta \in (0,1)$ be a failure {probability} and $R_t^{(i)}$ be the prediction error as in \eqref{eq:non_pred_err} for times $t\in\{1,\hdots,T\}$ and calibration trajectories $i\in\{1,\hdots,n_1\}$ from $D_{cal,1}$. Let $\alpha_1,\hdots,\alpha_T$ be obtained by solving \eqref{eq:originalProblemMILP} or \eqref{eq:relaxedProblemLP}. Further, let $R^{(i)}$ be the nonconformity score as in \eqref{eq:cal_R_} for calibration trajectories $i\in\{n_1+1,\hdots,n\}$ from $D_{cal,2}$. Finally, let $R=\max\left(\alpha_1 R_p(Y_1, \hat{Y}_1), \hdots, \alpha_T R_p(Y_T,\hat{Y}_T)\right)$. Then
\begin{align}\label{eq:cpGuarantee}
    &\text{Prob}\left(R \leq C \right) \geq 1-\delta
\end{align}
where $C:=\text{Quantile}(\{ R^{(n_1+1)}, \hdots, R^{(n)},\infty \},1-\delta)$.

\end{theorem}

We point out the necessity to split the dataset $D_{cal}$ into the datasets $D_{cal,1}$ and $D_{cal,2}$ to conform to the exchangeability assumption in conformal prediction. To convert equation \eqref{eq:cpGuarantee} into practical conformal prediction regions, note that 
\begin{align*}
     R=&\text{max}\left(\alpha_1 \|Y_1-\hat{Y}_1\|,\hdots,\alpha_T \|Y_T-\hat{Y}_T\| \right)\le C \\
    \Leftrightarrow\;\;\;\;\;\;\;\; &  \|Y_t-\hat{Y}_t\|\le C/\alpha_t \; \; \forall t\in\{1,\hdots,T\}
\end{align*}

Intuitively, the conformal prediction regions provide geometric balls of radius $C/\alpha_t$ around each prediction $\hat{y}_t$ due to the use of the Euclidean norm. For an example, we refer to \Cref{fig:ORCATrace}. We remark that other choices of the prediction $R_p$ may lead to different shapes. We summarize this result next. 

\begin{corollary}
Let all conditions of Theorem \ref{thm:2} hold. Then
    \begin{align*}
        \text{Prob}(\|Y_t-\hat{Y}_t\| \le C/\alpha_t \; \; \forall t\in\{1,\hdots,T\}) \geq 1-\delta.
    \end{align*}
\end{corollary}


\section{Case Study}
\label{sec:exps}

We consider two case studies to evaluate our approach. The first case study is about predicting pedestrian paths, while the second case study is about predicting the altitude of an F16 fighter jet performing ground avoidance maneuvers. We compare our method to the approach from \cite{stankeviciute2021conformal,lindemann2022safe,lindemann2023conformal} that is described in \Cref{sec:intro_cp}, which requires conservative union bounding to obtain valid prediction regions. We colloquially refer to this method as the union bound approach. We use Gurobi \cite{gurobi} to solve the optimization problems on a Windows machine running an Intel i7-855OU CPU with 4 cores and 16GB of RAM. The code for our case studies can be found on \url{https://github.com/earnedkibbles58/timeParamCPScores}.

Throughout this section, we use \eqref{eq:relaxedProblemLP} to compute the $\alpha$ values and we set $\delta:=0.05$.

\subsection{ORCA Pedestrian Simulator}
\label{exp:ORCA}

For \textcolor{black}{the first} case study, we analyze pedestrian location predictions using data generated from the ORCA simulator \cite{van2008reciprocal}. We use the social LSTM \cite{alahi2016social} architecture to make predictions up to $2.5$ seconds into the future at a rate of 8Hz (so that $T=20$ predictions in total) {with the previous $2.5$ seconds as input (so that $T_{obs}=20$)} using the trajnet++ framework \cite{kothari2021human}. For the exact details of training, we refer the reader to \cite{lindemann2022safe}.

We collect a dataset of $1291$ trajectories. We randomly select $646$ of these points for $D_{cal}$ and use the rest to form $D_{val}$, which we use for validation. We randomly select $50$  trajectories from $D_{cal}$ for $D_{cal,1}$  (to compute the $\alpha$ values) and the remaining trajectories for $D_{cal,2}$ (to compute the conformal prediction regions). We first show the results for one realization of this data sampling process and then we run $100$ trials to statistically analyze the results.

\subsubsection{Statistical Analysis}
\label{sec:caseStudyManyTrial}

For one realization, trajectories of the actual and predicted pedestrian locations are shown alongside the conformal prediction regions in \Cref{fig:ORCATrace}, the $\alpha$ values are shown in Appendix \ref{app:alphaVallORCA}, and scatter plots of the prediction errors over $D_{val}$ are shown in Appendix \ref{app:ORCAScatterPlots}.

Over $100$ trials, the average (trajectory level) coverage for our approach was $0.9558$ and the average coverage for the union bound approach was $0.9945$. Histograms of the validation coverage for the two approaches are shown in \Cref{fig:ORCAHists}. The union bound approach has larger coverage due to the conservatism introduced by union bounding. 
To further quantify the how much less conservative our prediction regions are, we also computed the average size of the conformal regions for each time horizon, which are shown in \Cref{fig:ORCACPAreas}. This highlights that our method produces valid conformal prediction regions that are significantly smaller than previous approaches.

\begin{figure}[h!]
    \centering
    \begin{subfigure}[]{0.48\columnwidth}
    \centering
    \includegraphics[width=0.9\linewidth]{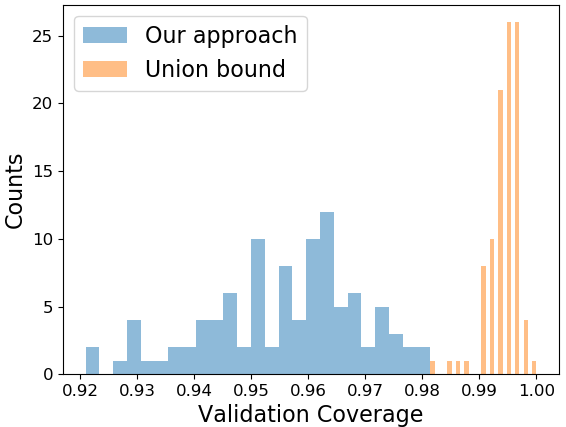}
        \caption{Validation coverage histograms }
        \label{fig:ORCAHists}
    \end{subfigure}
    \begin{subfigure}[]{0.48\columnwidth}
        \centering
        \includegraphics[width=0.9\linewidth]{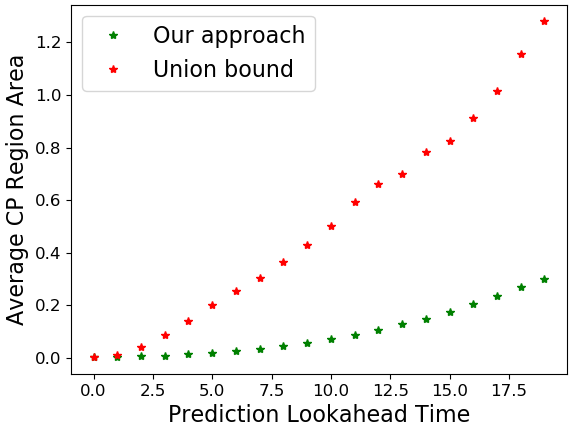}
        \caption{Average area of the conformal prediction regions }
        \label{fig:ORCACPAreas}
    \end{subfigure}
    \caption{Comparison of our approach and the union bound approach over $100$ trials for the ORCA case study.}
    \label{fig:ORCAFigs}
\end{figure}

\subsubsection{Original v. Relaxed Problem Difference}
\label{sec:alphaDiffs}

For each trial we also computed the objective function and $\alpha$ values using \Cref{eq:originalProblemMILP} to compare to those computed using \Cref{eq:relaxedProblemLP}. The average difference between the objective function values is $1.847*10^{-9}$. The average $l^2$ norm difference between the $\alpha$ values is $3.495*10^{-7}$ and the largest $l^2$ norm difference is $3.453*10^{-5}$. These differences are well within the solver's default tolerance of $10^{-4}$, which empirically validates the equivalence from \Cref{thm:1}. The average runtime for solving the original problem was $0.525$ seconds and the relaxed problem was $0.207$ seconds.

\subsection{F16 Fighter Jet}
\label{exp:f16}

For the second case study, we analyze the F16 fighter jet ground avoidance maneuver from the open source simulator \cite{Heidlauf2018}. We again use an LSTM to predict the altitude of the F16 fighter jet up to $2.5$ seconds into the future at a rate of 10Hz (so that is $T=25$ predictions in total) with the altitude from the past $2.5$ seconds as input. The trajectory predictor consists of an LSTM with two layers of width $25$ neurons each followed by two Relu layers with 25 neurons each and a final linear layer. We used $1500$ trajectories, each 5 seconds in duration, to train the network.

We collect a dataset of $1900$ trajectories, all of length 5 seconds. We randomly select $1000$ trajectories for $D_{cal}$ and the rest for $D_{val}$. We randomly select either $100$, $200$ or $400$ trajectories from $D_{cal}$ for $D_{cal,1}$ and the remaining trajectories in $D_{cal}$ for $D_{cal,2}$. We run $100$ trials to statistically analyze the approach, just as in \Cref{exp:ORCA}.

\subsubsection{Statistical Analysis}
\label{sec:f16Statistical}

The average (trajectory level) coverage for our approach was $0.953$ for each of $|D_{cal,1}|=400,200,100$ and the average coverage for the union bound approach was $0.993$. Histograms of the validation coverages for the union bound approach and our approach with $|D_{cal,1}|=400$ are shown in \Cref{fig:f16Hists}. The union bound approach has larger coverage due to the conservatism introduced by union bounding. For one realization trajectories of the actual and predicted altitudes along with the conformal prediction regions for our approach using $|D_{cal,1}|=400$ and the union bound approach are shown in Appendix \ref{app:f16SingleTrial}.

The variance of the coverage for our approach using $|D_{cal,1}|=400,200,100$ was $9.04*10^{-5}$, $7.63*10^{-5}$, and $5.89*10^{-5}$, while the variance of the coverage for the union bound approach was $1.51*10^{-5}$. For our approach, note that using larger $D_{cal,1}$ results in slightly higher variance for the coverage. That is because one is left with less data in $D_{cal,2}$ for running the conformal prediction step.

The average width of the conformal regions for each time horizon are shown in \Cref{fig:f16RegionWidths}. This demonstrates that our approach results in smaller prediction regions. Additionally, note that the prediction regions get slightly smaller as one uses more data for $D_{cal,1}$. That is because more data allows for selecting better $\alpha$ values. This highlights that there is a tradeoff in how much data one uses in $D_{cal,1}$ as opposed to $D_{cal,2}$. Making $D_{cal,1}$ larger will result in slightly smaller prediction regions, but the coverage of the regions will have higher variance, due to the smaller amount of data available in $D_{cal,2}$ for running the conformal prediction step.

\begin{figure}[h!]
    \centering
    \begin{subfigure}[]{0.48\columnwidth}
    \centering
    \includegraphics[width=\columnwidth]{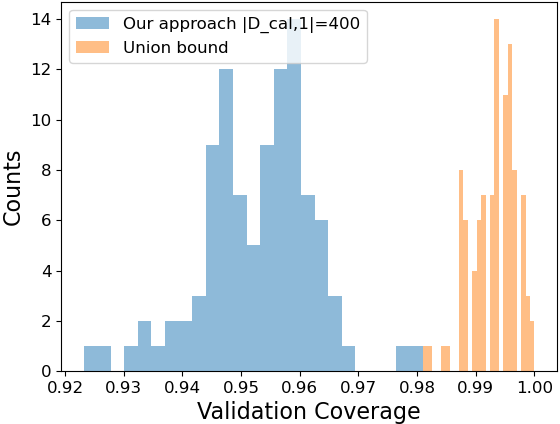}
        \caption{Validation coverage histograms for the union bound approach and our approach. }
        \label{fig:f16Hists}
    \end{subfigure}
    \begin{subfigure}[]{0.48\columnwidth}
    \centering
    \includegraphics[width=\columnwidth]{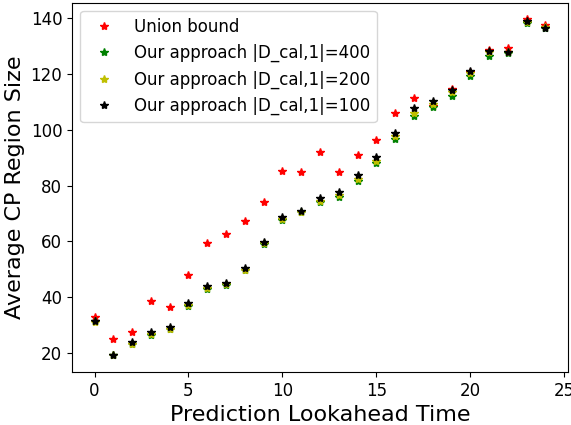}
        \caption{Average conformal prediction region width for our and the union bound approach. }
        \label{fig:f16RegionWidths}
    \end{subfigure}
    \caption{Comparison of our approach and the union bound approach over $100$ trials for the F16 case study.}
    \label{fig:F16Figs}
\end{figure}

\subsubsection{Original v. Relaxed Problem Difference}
\label{sec:f16AlphaDiffs}
We also computed the $\alpha$ values using \Cref{eq:originalProblemMILP} to compare to those computed using \Cref{eq:relaxedProblemLP}. In this case study, they were exactly the same, which empirically validates \Cref{thm:1}.

\subsection{Case Study Comparison}
\label{sec:caseStudyComparison}

Our method produced bigger relative improvements over the union bounding approach for the ORCA case study than for the F16 case study, see Figs. \ref{fig:ORCACPAreas} and \ref{fig:f16RegionWidths}. We hypothesize that this is due to the distributions of the non-conformity scores having longer tails in the ORCA case study, see histograms of individual prediction errors for each case study in Appendix \ref{app:NonConformityScoreQunatiles}. This explains the larger difference in the ORCA case study between our approach, which looks at the $1-\delta$ quantile, and the union bounding approach, which looks at the $1-\delta/T$ quantile. The $1-\delta$ and $1-\delta/T$ quantiles for the two case studies are shown in \Cref{fig:errQuantilesf16_}.  The difference in quantile values is much larger for ORCA than the F16 case study, confirming our hypothesis.

\begin{figure}[h!]
    \centering
    \begin{subfigure}[]{0.48\columnwidth}
        \centering
        \includegraphics[width=0.9\linewidth]{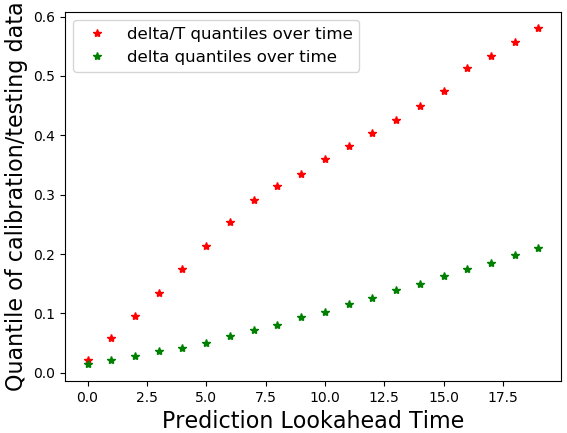}
        \caption{Pedestrian case study} 
        \label{fig:errQuantilesORCA}
    \end{subfigure}
    \begin{subfigure}[]{0.48\columnwidth}
        \centering
        \includegraphics[width=0.9\linewidth]{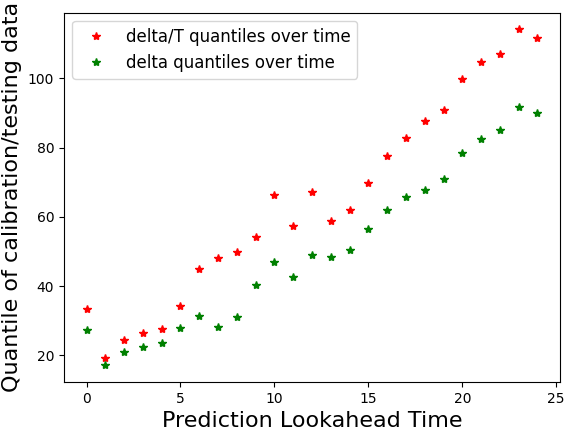}
        \caption{F16 case study} 
        \label{fig:errQuantilesf16}
    \end{subfigure}
    \caption{Non-conformity score quantiles for the calibration and testing data for the two case studies.}
    \label{fig:errQuantilesf16_}
\end{figure}


\section{Conclusion}
\label{sec:conclusion}

In this paper, we have presented a method for producing conformal prediction regions for time series data that are significantly tighter than previous approaches. To do this, we define a parameterized non-conformity score function and use optimization to fit the function to calibration data. This allows us to use standard inductive conformal prediction to get valid prediction regions. For future work, we plan to integrate our conformal prediction approach into planning and control frameworks and apply parameterized non-conformity score functions to the non-time series setting.

\section*{Acknowledgements}
This work was generously supported by NSF award SLES-2331880.

\bibliography{literature}

\onecolumn
\appendix

\section{Proof of \Cref{thm:alphasNonZero}}
\label{app:alphasNonzeroProof}

\begin{proof}[Proof of Theorem \ref{thm:alphasNonZero}]

 We argue that $\alpha_1^*, \hdots,\alpha_T^*>0$  by contradiction. Let therefore $\alpha_1, \hdots,\alpha_T$  be a feasible (not necessarily optimal) solution to the optimization problem in \eqref{eq:highLevelOptimization_} with $\alpha_t=0$ for some $t\in\{1,\hdots,T\}$. 



First, note that at least one $\alpha_1,\hdots,\alpha_T$ must be positive due to constraint \eqref{eq:highLevelOptimizationc}, i.e., there exists $\alpha_\tau>0$ for some $\tau\neq t$. Because $R_{t}^{(i)}>0$ for all $t\in\{1,\hdots,T\}$ and for all $i\in\{1,\hdots,n_1\}$ by assumption, i.e., no $R_t^{(i)}$ is zero, and there exists $\alpha_\tau>0$ for some $\tau\neq t$, it follows that $$ \max(\alpha_1 R_1^{(j)},\hdots,\alpha_T R_T^{(j)}) > 0.$$

Consequently, we know that
\begin{align}
    0=\alpha_t R_t^{(i)} < \max(\alpha_1 R_1^{(i)},\hdots,\alpha_T R_T^{(i)}) = R^{(i)} \; \; \forall i \in \{1,\hdots,n_1\}.
\end{align}

Now, since the maximum is a continuous function and the Quantile is a continuous, monotone function with respect to the $R^{(i)}$ values \footnote{Note that the Quantile function is not continuous with respect to the $1-\delta$ input, but that input remains constant in this case.}, we can decrease the value of the objective function in \eqref{eq:highLevelOptimization} by increasing the value of $\alpha_t$ (from 0) while decreasing the values of $\alpha_\tau$ with $\tau \neq t$. 
By contradiction, it holds that the optimal value $\alpha_t^*$ will be non-zero, i.e, $\alpha_t^*>0$. 

\end{proof}

\section{Proof of \Cref{thm:1}}
\label{app:thm1}

\begin{proof}[Proof of Theorem \ref{thm:1}]
    Assume that $\alpha_1^*,\hdots,\alpha_T^*,R^{(1),*},\hdots,R^{(n_1),*}$ is an optimal solution to the optimization problem in equation \eqref{eq:problemWithoutBinaryVars}. We will first show that we can use $\alpha_1^*,\hdots,\alpha_T^*,R^{(1),*},\hdots,R^{(n_1),*}$ to construct a feasible solution to the optimization problem in \eqref{eq:highLevelOptimizationMaxMILP}, which we denote by $\bar{\alpha}_1^*,\hdots,\bar{\alpha}_T^*,\bar{R}^{(1),*},\hdots,\bar{R}^{(n_1),*}$, without altering the optimal value of the objective function. Particularly, let
    \begin{subequations}
    \begin{align}
        \bar{\alpha}_t^* &:= \alpha_t^*, \; t=1,\hdots,T \label{eq:exSolnMILPAlphas} \\
        \bar{R}^{(i),*} &:= \text{max}(\bar{\alpha}^*_1 R^{(i)}_1, \hdots, \bar{\alpha}^*_T R^{(i)}_T), \; i=1,\hdots,n_1 \label{eq:exSolnMILPRs}
    \end{align}
    \end{subequations}

    Note that there exists binary variables $\bar{b}_t^{(i)}$ along with which $\bar{\alpha}_1^*,\hdots,\bar{\alpha}_T^*,\bar{R}^{(1),*},\hdots,\bar{R}^{(n_1),*}$ are a feasible solution to the optimization problem in \eqref{eq:highLevelOptimizationMaxMILP}. This follows since \eqref{eq:exSolnMILPAlphas} satisfies \eqref{eq:highLevelOptimizationc} and \eqref{eq:highLevelOptimizationd} and \eqref{eq:exSolnMILPRs} satisfies \eqref{eq:highLevelOptimizationb} and consequently (\ref{eq:highLevelOptimizationMaxMILPb}-\ref{eq:highLevelOptimizationMaxMILPd}) for appropriate choices of $\bar{b}_t^{(i)}$ (and $M$). 

    Note specifically that $\bar{\alpha}_t^*$ and $\alpha_t^*$ are the same, while this does not necessarily hold for ${R}^{(i),*}$ and $\bar{R}^{(i),*}$. However, it is guaranteed that $\bar{R}^{(i),*} \leq R^{(i),*} \; \forall i=1,\hdots,n_1$ as the constraint in \eqref{eq:problemWithoutBinaryVarsb} ensures that $R^{(i),*}$ is lower bounded by $\text{max}(\alpha^*_1 R^{(i)}_1, \hdots, \alpha^*_T R^{(i)}_T)=\text{max}(\bar{\alpha}^*_1 R^{(i)}_1, \hdots, \bar{\alpha}^*_T R^{(i)}_T)$. Since the Quantile function is a monotone operator (its value is non-increasing if inputs are non-increasing), the value of the objective function in \eqref{eq:highLevelOptimizationMaxMILP} that corresponds to $\bar{\alpha}_1^*,\hdots,\bar{\alpha}_T^*,\hdots,\bar{R}^{(1),*},\bar{R}^{(n_1),*}$ cannot be greater than the value of the objective function in \eqref{eq:problemWithoutBinaryVars} that corresponds to $\alpha_1^*,\hdots,\alpha_T^*,R^{(1),*},\hdots,R^{(n_1),*}$.

    At the same time, we note that the optimal value of the optimization problem in \eqref{eq:problemWithoutBinaryVars} cannot be greater than the optimal value of the optimization problem in \eqref{eq:highLevelOptimizationMaxMILP} because it has strictly fewer constraints than \eqref{eq:problemWithoutBinaryVars} so that the feasible set of \eqref{eq:problemWithoutBinaryVars} contains the feasible of \eqref{eq:highLevelOptimizationMaxMILP}. As a consequence, the optimal values of \eqref{eq:highLevelOptimizationMaxMILP} and \eqref{eq:problemWithoutBinaryVars} are equivalent.

\end{proof}

\section{KKT Conditions of \eqref{eq:quantileAsLP}}
\label{app:KKTConds}

The KKT conditions of the LP in \eqref{eq:quantileAsLP} are:
\begin{subequations} \label{eq:KKTconds}
\begin{align}
    & (1-\delta) - u_i^+ + v_i = 0, \; \; i=1,\hdots,n_1 \label{eq:KKTcondsa} \\
    & \delta - u_i^- - v_i = 0, \; \; i=1,\hdots,n_1 \label{eq:KKTcondsb}  \\
    & \sum_{i=1}^{n_1} v_i = 0 \label{eq:KKTcondsc}  \\
    & u_i^+ e_i^+ = 0, \; \; i=1,\hdots,n_1 \label{eq:KKTcondsd}  \\
    & u_i^- e_i^- = 0, \; \; i=1,\hdots,n_1 \label{eq:KKTcondse}  \\
    & e_i^+ \geq 0, \; \; i=1,\hdots,n_1 \label{eq:KKTcondsf}  \\
    & e_i^- \geq 0, \; \; i=1,\hdots,n_1 \label{eq:KKTcondsg}  \\
    & e_i^+ + q - e_i^- - R^{(i)} = 0, \; \; i=1,\hdots,n_1 \label{eq:KKTcondsh}  \\
    & u_i^+ \geq 0, \; \; i=1,\hdots,n_1 \label{eq:KKTcondsi}  \\
    & u_i^- \geq 0, \; \; i=1,\hdots,n_1 \label{eq:KKTcondsj} 
\end{align}
\end{subequations}
Here, the constraints \eqref{eq:KKTcondsa}, \eqref{eq:KKTcondsb} and \eqref{eq:KKTcondsc} encode the stationarity condition, the constraints \eqref{eq:KKTcondsd} and \eqref{eq:KKTcondse} encode complementary slackness, the constraints \eqref{eq:KKTcondsf}, \eqref{eq:KKTcondsg}, and \eqref{eq:KKTcondsh} encode primal feasibility, and the constraints \eqref{eq:KKTcondsi} and \eqref{eq:KKTcondsj} encode dual feasibility.

\section{Proof of Theorem \ref{thm:2}}
\label{app:thm2}
\begin{proof}[Proof of Theorem \ref{thm:2}]
    Since $\alpha_1,\hdots,\alpha_T$ are obtained from $D_{cal,1}$, the non-conformity scores $R^{(i)}$ which are computed from $D_{cal,2}$ are exchangeable. Consequently, we can apply \cite[Lemma 1]{tibshirani2019conformal} to get that \Cref{eq:cpGuarantee} holds.
\end{proof}

\section{Alpha Values from ORCA Case Study}
\label{app:alphaVallORCA}

The $\alpha$ values for one realization of the ORCA case study are shown in \Cref{tab:alphaVals}. They decrease as the prediction horizon increases, which is what we expect as the prediction error generally increases over time.

\begin{table}[h]
\caption{Values of $\alpha$ parameters for ORCA case study.}
    \centering
    \begin{tabular}{|c|c|c|c|}
    \hline Time step  & $\alpha$ &  Time step  & $\alpha$  \\ \hline
    1 & 0.21134 & 11 & 0.0267 \\ \hline
    2 & 0.1495 & 12 & 0.0229 \\ \hline
    3 & 0.1057 & 13 & 0.0207 \\ \hline
    4 & 0.0945 & 14 & 0.0189 \\ \hline
    5 & 0.0668 & 15 & 0.0171 \\ \hline
    6 & 0.0555 & 16 & 0.0158 \\ \hline
    7 & 0.0423 & 17 & 0.0148 \\ \hline
    8 & 0.0371 & 18 & 0.0138 \\ \hline
    9 & 0.0330 & 19 & 0.0125 \\ \hline
    10 & 0.0290 & 20 & 0.0118 \\ \hline
    \end{tabular}
    \label{tab:alphaVals}
\end{table}

\section{ORCA Case Study Scatter Plots}
\label{app:ORCAScatterPlots}

Scatter plots of the prediction errors over $D_{val}$ and the conformal prediction region radii for one realization of the ORCA case study are shown in \Cref{fig:ORCAErrScatter}.

\begin{figure*}[h!]
    \centering
    \begin{subfigure}[b]{0.32\textwidth}
        \centering
        \includegraphics[width=0.9\textwidth]{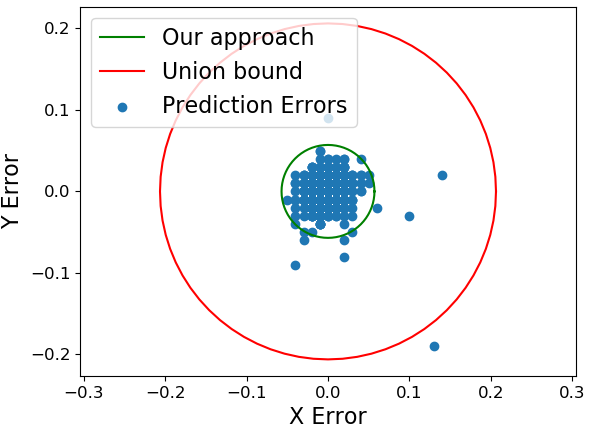}
        \caption{5 step prediction error} 
        \label{fig:4stepAhead}
    \end{subfigure}
    \begin{subfigure}[b]{0.32\textwidth}
        \centering
        \includegraphics[width=0.9\linewidth]{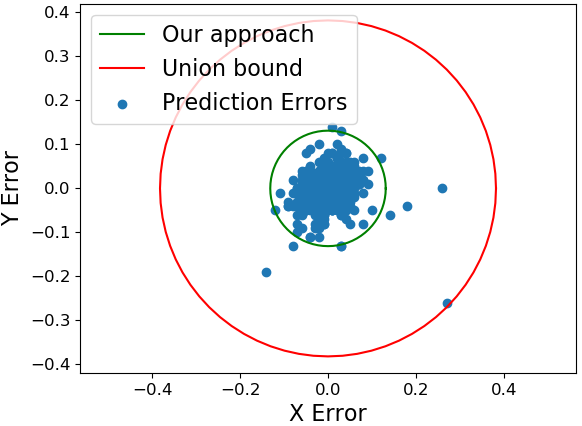}
        \caption{10 step prediction error} 
        \label{fig:10stepAhead}
    \end{subfigure}
    \begin{subfigure}[b]{0.32\textwidth}
        \centering
        \includegraphics[width=0.9\linewidth]{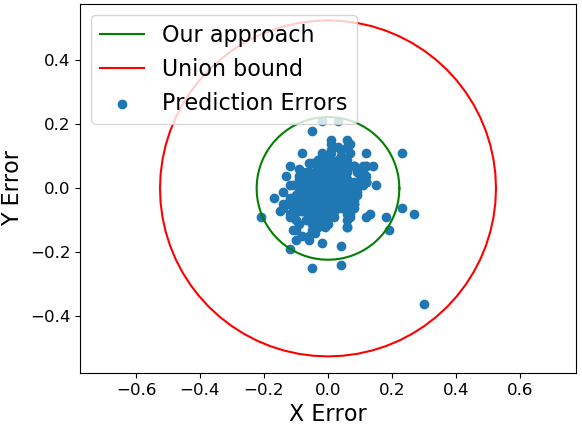}
        \caption{15 step prediction error} 
        \label{fig:15stepAhead}
    \end{subfigure}
    \caption{Conformal prediction regions and errors for various prediction horizons. The green circles show the regions produced by our approach and the red circles show the regions produced by the union bound approach. Notably, our approach produces much tighter valid predictions regions.}
    \label{fig:ORCAErrScatter}
\end{figure*}

\section{F16 Single Trial}
\label{app:f16SingleTrial}
For one realization of the F16 case study, a trajectory of the actual and predicted F16 altitudes, along with the conformal prediction regions for our approach using $|D_{cal,1}|=400$ and the approach from \cite{lindemann2022safe} is shown in \Cref{fig:f16Trace}. 

\begin{figure}[h!]
    \centering
    \includegraphics[width=0.33\textwidth]{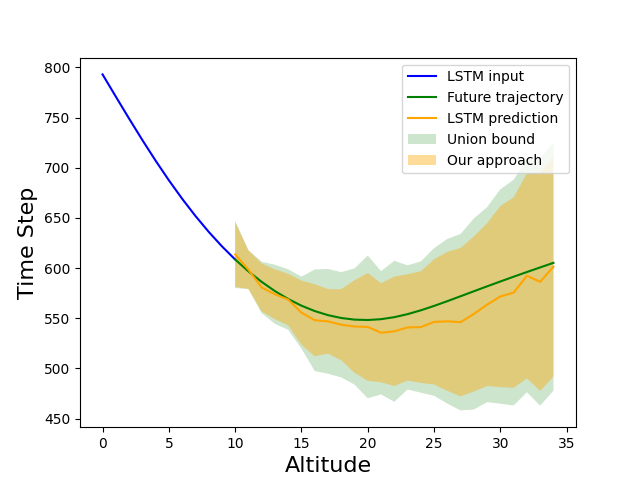}
        \caption{Sample trajectory of an F16 fighter jet (blue and green lines), the future predictions using an LSTM (orange line), the conformal prediction regions for our approach with $|D_{cal,1}=400|$ (orange shaded region), and the conformal prediction regions for the approach from \cite{lindemann2022safe} (green shaded region).}
        \label{fig:f16Trace}
\end{figure}

\section{Non-conformity Score Quantiles for Case Studies}
\label{app:NonConformityScoreQunatiles}

Histograms of individual prediction errors for each case study are shown in \Cref{fig:errHists}.

\begin{figure*}[h!]
    \centering
    \begin{subfigure}[b]{0.33\textwidth}
        \centering
        \includegraphics[width=\textwidth]{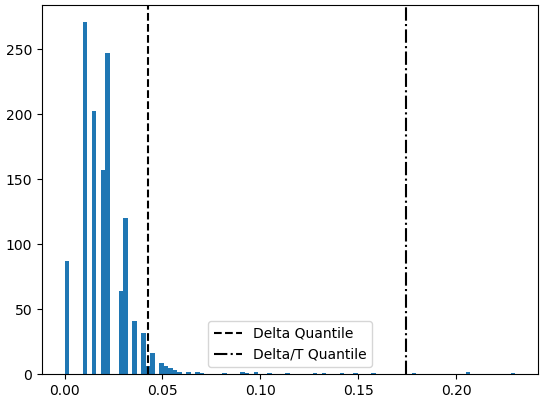}
        \caption{Non-conformity score histogram for one prediction time step of the calibration and testing data from the pedestrian prediction case study }
        \label{fig:errHistORCA}
    \end{subfigure}
    \begin{subfigure}[b]{0.33\textwidth}
        \centering
        \includegraphics[width=\linewidth]{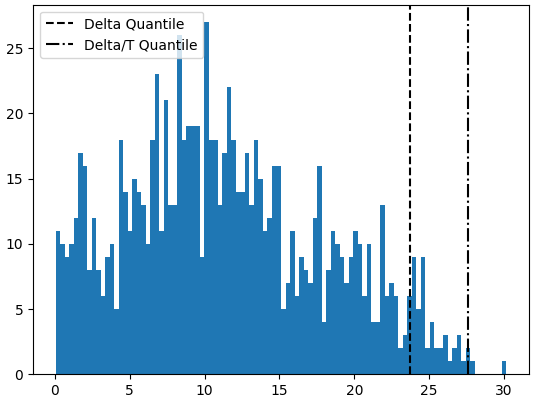}
        \caption{Non-conformity score histogram for one prediction time step of the calibration and testing data from the F16 case study} 
        \label{fig:errHistf16}
    \end{subfigure}
    \caption{Non-conformity score quantiles for the calibration and testing data from the two case studies.}
    \label{fig:errHists}
\end{figure*}

\clearpage

\end{document}